\documentclass[conference]{IEEEtran}

\usepackage{graphics, array, url, hyperref, epsfig, amsmath, amssymb, graphicx}
\usepackage{subfigure}
\usepackage{threeparttable}
\usepackage{xcolor}
\usepackage{slashbox}
\usepackage{ wasysym }
\usepackage{setspace}
\usepackage{amsthm}

\newtheorem{prop}{Proposition}
\newtheorem{lemma}{Lemma}
\DeclareMathOperator*{\argmax}{arg\,max}
\usepackage[ruled,vlined,linesnumbered,noresetcount]{algorithm2e}

\begin{document}

\title{Compensation of Charging Station Overload via On-road Mobile Energy Storage Scheduling}

\author{\IEEEauthorblockN{Nan Chen$^{\dagger}$, Mushu Li$^{\dagger}$, Miao Wang$^{\star}$, Jinghuan Ma$^{+}$ and Xuemin (Sherman) Shen$^{\dagger}$}
\IEEEauthorblockA{$^{\dagger}$Department of Electrical and Computer Engineering, University of Waterloo, Waterloo, Ontario, Canada} 
\IEEEauthorblockA{$^\star$Department of Electrical and Computer Engineering, Miami University, Oxford, Ohio, United States} 
\IEEEauthorblockA{$^{+}$School of Electronic Information and Electrical Engineering, Shanghai Jiao Tong University, Shanghai, China}
 Email: n37chen@uwaterloo.ca, m475li@uwaterloo.ca, wangm64@miamioh.edu, mjhdtc@sjtu.edu.cn, sshen@uwaterloo.ca}


\maketitle

\begin{abstract}

Supported by the technical development of electric battery and charging facilities, plug-in electric vehicle (PEV) has the potential to be mobile energy storage (MES) for energy delivery from resourceful charging stations (RCSs) to limited-capacity charging stations (LCSs). In this paper, we study the problem of using on-road PEVs as MESs for energy compensation service to compensate charging station (CS) overload. A price-incentive scheme is proposed for power system operator (PSO) to stimulate on-road MESs fulfilling energy compensation tasks. The price-service interaction between the PSO and MESs is characterized as a one-leader, multiple-follower Stackelberg game. The PSO acts as a leader to schedule on-road MESs by posting service price and on-road MESs respond to the price by choosing their service amount. The existence and uniqueness of the Stackelberg equilibrium are validated, and an algorithm is developed to find the equilibrium. Simulation results show the effectiveness of the proposed scheme in utility optimization and overload mitigation.
\let\thefootnote\relax\footnote{\emph{This Manuscript has been accepted by IEEE Globecom'19.}}
\let\thefootnote\relax\footnote{\emph{Copyright (c) 2015 IEEE. Personal use of this material is permitted. However, permission to use this material for any other purposes must
be obtained from the IEEE by sending a request to pubs-permissions@ieee.org.}}

\end{abstract}

\section{Introduction}

Recently, the advancement of electric battery technology pushes forward the prevalence of plug-in electric vehicles (PEVs) in the automobile market \cite{expectation}. The increasing PEV charging demand, especially at peak hours, puts great pressure on the charging stations (CSs) that have limited charging capacities. The potential overload at CS feeders could incur severe power quality issues and transformer degradation \cite{impact2}. Therefore, additional power supply is required by infrastructure upgrade such as deploying flexible energy storages. Thanks to the experimental success of the bi-directional charger, the PEV can transmit energy from its rechargeable battery to the power system as a mobile energy storage (MES) \cite{bidirectional, nanTVT}. When limited-capacity CSs (LCSs) encounter power shortages in peak hours, MESs can deliver energy from CSs with redundant power (\emph{i.e.}, resourceful CSs (RCSs)) to LCSs. By fully utilizing the system energy resource, potential overload issues at LCSs can be mitigated without excessive infrastructure upgrade expenditure. As the PEV commercialization proceeds, and PEV becomes one of regular transportation options, a considerable portion of on-road PEVs can be stimulated to serve as MESs to compensate the LCS overload.

In the literature, many research works study MES scheduling to balance the power supply and demand effectively. The work in \cite{pyi} utilizes MESs to transmit energy from renewable energy plants to CSs and minimizes the MES transmission loss. MESs can also provide demand response service to charge or discharge energy depending on user needs as in the work \cite{DSMrobustness}. In our previous work \cite{nanTVT}, we propose to use MESs that belong to the PSO to fulfill energy compensation tasks among a group of CSs (GCS). Most related works consider that MESs fully comply with PSO commands. However, the MESs we schedule in this work are private-owned. These MESs have their own travelling plans, and the energy compensation service is considered as an additional on-road task rather than a mandatory task. Therefore, the PSO needs to provide additional incentive for these MESs to accomplish the overload compensation task.

As a well-developed mathematical model, game theory can precisely characterize incentive interactions between MESs and the PSO. Specifically, the PSO first posts the service price of energy compensation and then, in response to the posted price, MESs decide their service amount. This price-service interaction can be formulated by a sequential game model, such as the Stackelberg game for the interaction analysis. The Stackelberg game model has been applied in PEV charging scheduling \cite{approach}-\cite{multicharging}. In the work \cite{approach}, the CS is considered as the leader to maximize its charging revenue while PEVs are considered as followers to maximize their charging energy fairly. When scheduling PEV charging among a GCS, the GCS can be managed together by the PSO to maximize the overall energy utilization and revenue as in the work \cite{2dmarkov}, or CSs can compete with each other to form a multi-leader multi-follower game as in the work \cite{multicharging}. 
Stackelberg game can also be applied for PEV discharging as in the work \cite{v2gglobecom}, where the CS adjusts its price to maximize its charging and vehicle-to-grid service revenue. While PEV charging and discharging scheduling has been explored in the above works, their primary focus is on PEV in-station scheduling. However, in the case of MES scheduling, the service process becomes much more complicated. For MESs to accomplish energy compensation tasks, they are motivated to charge and discharge at the CSs along their travel routes. Therefore, the costs of additional battery degradation and service time should be considered. 

In this paper, we consider a scenario where the PSO stimulates on-road MESs by providing service incentive to mitigate the LCS overload. Our main contributions are as follows:

\begin{itemize}

\item A price-incentive scheme is proposed to stimulate MESs participating in the service by increasing their service revenues. The proposed scheme also guarantees a cost-efficient overload mitigation from the PSO perspective.

\item  A Stackelberg game is formulated to characterize the interaction between the PSO and MESs, where the PSO acts as the leader and MESs act as followers. The existence and uniqueness of the Stackelberg equilibrium are validated, and an algorithm is designed to find the equilibrium.

\end{itemize}

The remainder of the paper is organized as follows. The system model is introduced in Section II. The Stackelberg game is formulated in Section III, followed by the game analysis in Section IV. The simulation results are presented in Section V. Finally, the conclusion is given in Section VI.

\section{System Model}

As shown in Fig. \ref{Fig_systemmodel}, the system model consists of a GCS, on-road MESs, the PSO, and communication infrastructures. 
The system model analysis time window $\mathcal{H}$ is partitioned into $H$ time slots with equal interval of $\Delta t$. Consider the MES energy compensation service is fulfilled within $\Delta t$, and the MES scheduling is regularly conducted at each time slot $h \in \mathcal{H}$.
We consider that a GCS is composed of RCSs and LCSs that are geographically reachable by vehicles. RCSs, a set of CSs denoted by \begin{math} \mathcal{R}=\{R_\mathrm{1},R_\mathrm{2}\dots R_\mathrm{i}\} \end{math}, are normally deployed at urban areas with sufficient power supplies. In addition to charging arriving PEVs with charging demands, the surplus energy of RCSs can be stored by MESs and delivered to LCSs. LCSs, a set of CSs denoted by \begin{math} \mathcal{L}=\{L_\mathrm{1},L_\mathrm{2}\dots L_\mathrm{j}\} \end{math}, are usually deployed at rural areas with limited power capacities and thus could encounter overload issues at peak hours. 
At time $h$, RCS $R_\mathrm{i}$ sends information of its surplus energy $E_\mathrm{i,h}$ to the PSO while LCS $L_\mathrm{j}$ sends information of its minimal demanding energy $D^L_\mathrm{j,h}$ and maximal demanding energy $D^U_\mathrm{j,h}$ to the PSO. Both information is sent via wired communication technology such as fiber optic.

\begin{figure}[!t]
\setlength{\abovecaptionskip}{-5pt}
\centering
\includegraphics[width=0.5\textwidth]{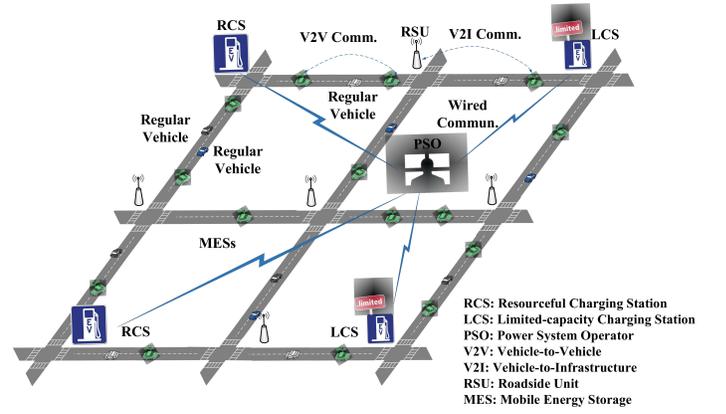}
\vspace{-0.3cm}
\caption{System model.}
\vspace{-0.3cm}
\label{Fig_systemmodel}
\end{figure}

A set of on-road MESs, denoted as \begin{math}\mathcal{K}=\{1, \dots k, \dots K \} \end{math}, can serve the energy compensation tasks when their planned travel routes pass the PSO's targeted RCSs and LCSs. Upon receiving service requests, on-road MESs send information of their planned travel routes and energy compensation capacities to the PSO. As MESs are constantly moving along the road, mobile-support wireless communication technology can be adopted for the information exchange between MESs and the PSO. For example, vehicle ad-hoc networks (VANETs) can be adopted to transmit the vehicle information to the PSO through vehicle-to-vehicle (V2V)/vehicle-to-infrastructure (V2I) communication.

Upon receiving information from MESs and CSs, the PSO first estimates the on-road MES service capacity. If an MES plans to travel from RCS $R_\mathrm{i}$ to LCS $L_\mathrm{j}$, the MES is counted as an energy compensation server for $R_\mathrm{i}-L_\mathrm{j}$ pair. By the end, PSO knows the number of MESs $N_\mathrm{ij}$ that can deliver energy from RCS $R_\mathrm{i}$ to LCS $L_\mathrm{j}$. Then, by analyzing the energy states of CSs, the PSO posts the service price to on-road MESs to stimulate them providing energy compensation service.

\section{Game Formulation}

In terms of time-variant GCS balance states, the interaction between the  PSO and MESs is formulated as a Stackelberg game at $h$-th time slot. As the MES scheduling is conducted at each time slot, the notation $h$ is omitted.

\subsection{Game Process}
 
We define the game in its strategic form: \begin{math} G=\{ \{\mathcal{K} \cup \{PSO\}\}, \{p\}, \{ e_\mathrm{k}\}_\mathrm{k\in \mathcal{K}}, \{U_\mathrm{P}\}, \{U_\mathrm{k}\}_\mathrm{k\in \mathcal{K}} \} \end{math}, where $\{ e_\mathrm{k}\}_\mathrm{k\in \mathcal{K}}$ denotes the set of strategies of MESs. $\{p\}$ denotes the PSO strategy (\emph{i.e.}, pricing); $\{U_\mathrm{P}\}$ and $\{U_\mathrm{k}\}_\mathrm{k\in \mathcal{K}}$ represent utility functions of the PSO and MESs, respectively. For a given service price $p$ by the PSO, the interaction between MESs is characterized as a non-cooperative game as follows:

\begin{itemize}
\item Players: the set of MESs $\mathcal{K}$.
\item Strategies: MES $k \in \mathcal{K}$, chooses an energy service amount $e_\mathrm{k}$.
\item Payoffs: the $k$-th MES receives utility $U_\mathrm{k}(e_\mathrm{k},p)$.
\end{itemize}

To find the Nash equilibrium, we need to find the best response function $e^\star_\mathrm{k}(p)$ of $k$-th MES under the service price $p$. The set of best  response functions $\{e^\star_\mathrm{k}(p)\}_\mathrm{k\in \mathcal{K}}$ is then sent to the PSO. PSO chooses the optimal service price $p^\star$ that maximizes its utility function $U_\mathrm{P}(p,\{e^\star_\mathrm{k}(p)\}_\mathrm{k\in \mathcal{K}})$.

\subsection{MES Model}

For MES \begin{math} k \in \mathcal{K} \end{math}, its utility function is defined as:

\begin{equation}
\label{MESutility}
U_\mathrm{k}(e_\mathrm{k},p)=pe_\mathrm{k}+p(e_\mathrm{k}-\bar{e})-\alpha^T_\mathrm{k}C^T_\mathrm{k}-\alpha^D_\mathrm{k}C^D_\mathrm{k}.
\end{equation}
 
\noindent The first term in equation (\ref{MESutility}) is the service reward calculated by multiplying the MES served energy amount with the service price. The second term is the motivation reward that motivates MES providing energy more than the expected average service amount $\bar{e}$. $\bar{e}$ is calculated by averaging the overall LCS minimal demands by the overall on-road MES number, denoted as \begin{math} \bar{e}=(\sum_{L_\mathrm{j}\in \mathcal{L}}D^L_\mathrm{j})/(\sum_\mathrm{R_\mathrm{i}\in \mathrm{R}}\sum_{L_\mathrm{j}\in \mathcal{L}}N_\mathrm{ij}) \end{math}. When the MES provides less energy than $\bar{e}$, the motivation reward is  negative, meaning that the MES receives less reward than the service reward. On the contrary, when the MES provides more service energy than $\bar{e}$, it will be rewarded more than service reward. The third term is the weighted service time cost that is the multiplication of the service time weight $\alpha^T_\mathrm{k}$ of MES $k$, and the service time $C^T_\mathrm{k}$ of MES $k$. The service time consists of MES charging time at its passing RCS $R_\mathrm{i}$ and discharging time at its destined LCS $L_\mathrm{j}$, which is denoted as:

 \begin{equation}
C^T_\mathrm{k}=\frac{e_\mathrm{k}}{P_\mathrm{i}}+\frac{e_\mathrm{k}}{P_\mathrm{j}},
 \end{equation}
 
 \noindent where $P_\mathrm{i}$ and $P_\mathrm{j}$ denote average charging/discharging power of $R_\mathrm{i}$ and $L_\mathrm{j}$ respectively. The travelling time is excluded from the utility function as $R_\mathrm{i}$ and $L_\mathrm{j}$ are on the MES planed travel route. The weight of service time cost $\alpha^T_\mathrm{k}$ indicates the MES driver preference towards service time and \begin{math} \alpha^T_\mathrm{k}>0 \end{math}. A high $\alpha^T_\mathrm{k}$ indicates that MES driver is unwilling to spend too much time in-station. The forth term of the function is the weighted battery degradation cost of MES discharging at LCS. It is calculated as the multiplication of the battery degradation weight $\alpha^D_\mathrm{k}$ of $k$-th MES and the battery degradation cost $C^D_\mathrm{k}$ of $k$-th MES. The battery degradation cost refers to a modified model as in the work \cite{batterydegradation}:
 
 \begin{equation}
C^D_\mathrm{k}=(\beta_\mathrm{1}P^3_\mathrm{j}+\beta_\mathrm{2}P^2_\mathrm{j}+\beta_\mathrm{3}P_\mathrm{j}+\beta_\mathrm{4})(\alpha_\mathrm{1}\frac{e^2_\mathrm{k}}{B^2_\mathrm{k}}+\alpha_\mathrm{2}\frac{e_\mathrm{k}}{B_\mathrm{k}}).
 \end{equation}
 
 \noindent The term of discharging power degradation, $\beta_\mathrm{1}P^3_\mathrm{j}+\beta_\mathrm{2}P^2_\mathrm{j}+\beta_\mathrm{3}P_\mathrm{j}+\beta_\mathrm{4}$, is a cubic function with coefficients $\beta_\mathrm{1}$, $\beta_\mathrm{2}$, $\beta_\mathrm{3}$, and $\beta_\mathrm{4}$. It is positively related to discharging power $P_\mathrm{j}$ at the MES destined LCS $L_\mathrm{j}$.  The term of depth-of-discharge (DoD) degradation, $\alpha_\mathrm{1}\frac{e^2_\mathrm{k}}{B^2_\mathrm{k}}+\alpha_\mathrm{2}\frac{e_\mathrm{k}}{B_\mathrm{k}}$, is a quadratic function that is positively related to the battery DoD $\frac{e_\mathrm{k}}{B_\mathrm{k}}$, and thus coefficient \begin{math} \alpha_\mathrm{1}>0 \end{math}. Similar to $\alpha^T_\mathrm{k}$, a high degradation weight $\alpha^D_\mathrm{k}$ denotes a high unwillingness to discharge. To simplify the equation, we denote \begin{math} D_\mathrm{k}=\beta_\mathrm{1}P^3_\mathrm{j}+\beta_\mathrm{2}P^2_\mathrm{j}+\beta_\mathrm{3}P_\mathrm{j}+\beta_\mathrm{4}>0 \end{math}.
 
Meanwhile, the MES service energy should be within its feasible range:

\begin{equation}
\label{MESconstraint}
0 \leq e_\mathrm{k} \leq B_\mathrm{k}-e^I_\mathrm{k}, 
\end{equation}  
 
\noindent where $B_\mathrm{k}$ denotes the $k$-th MES battery capacity and $e^I_\mathrm{k}$ denotes the initial state-of-charge (SoC) of MES $k$. MESs will not participate in the service when they cannot obtain any profit, and thus the utility function needs to satisfy:

\begin{equation}
\label{MESutilityC}
U_\mathrm{k}(e_\mathrm{k},p)>0.
\end{equation}
 
\noindent Therefore, given the posted price $p$, the MES decision making process is formulated as an optimization problem:

 \begin{align}
 \label{Op1}
\max_{e_\mathrm{k}} & \quad U_\mathrm{k}(e_\mathrm{k}, p) \\
 \textrm{s.t.}  & \quad (\ref{MESconstraint}), (\ref{MESutilityC}), \quad \forall k \in \mathcal{K} \notag.
\end{align}

 \subsection{PSO Model}
 
As the operator of the GCS, the PSO adjusts the posted price to maximize its utility function, which is denoted as:

\begin{multline}
\label{PSOutility}
U_\mathrm{P}(p,e_\mathrm{k})=\alpha_\mathrm{L}\sum_{L_\mathrm{j}\in \mathcal{L}}(-(a_\mathrm{j}\sum_{k \in \Lambda_\mathrm{j}}e_\mathrm{k}-b_\mathrm{j})^2+c_\mathrm{j}) \\
-(\sum_{k \in \mathcal{K}} pe_\mathrm{k}+\sum_{k \in \mathcal{K}} p(e_\mathrm{k}-\bar{e})).
\end{multline}

\noindent The first term is the weighted loading revenue that is the product of loading weight $\alpha_\mathrm{L}$ and the summation of LCS loading revenues. Denote a set of MESs whose destined LCS is $L_\mathrm{j}$ as $\Lambda_\mathrm{j}$. For LCS $L_\mathrm{j}$, the loading revenue increases as more energy delivered to the station and reaches the peak revenue at the maximal demanding load $D^U_\mathrm{j}$. Therefore, the loading revenue of LCS $L_\mathrm{j}$ is characterized as a quadratic function with its peak value at $D^U_\mathrm{j}$. We set \begin{math} a_\mathrm{j}=5\times 10^{-4} D^U_\mathrm{j}, b_\mathrm{j}=a_\mathrm{j}D^U_\mathrm{j}, c_\mathrm{j}=b_\mathrm{j}^2\end{math}. The second term of the function is the summation of all MES service costs and motivation costs, as introduced in Section III-B.

For LCS $L_\mathrm{j}$, the MES delivered energy should be within its demanding energy range \begin{math} [D^L_\mathrm{j}, D^U_\mathrm{j}] \end{math}. Thus, the LCS energy constraint is denoted as:

\begin{equation}
\label{PSOC1}
D^L_\mathrm{j} \leq \sum_{k \in \Lambda_\mathrm{j}} e_\mathrm{k} \leq D^U_\mathrm{j}, \forall L_\mathrm{j} \in \mathcal{L}.
\end{equation}

\noindent On the energy supplier side, the energy stored by MESs cannot exceed the maximal surplus energy capacities at RCSs. Denote a set of MESs whose passing RCS is $R_\mathrm{i}$ as $\Omega_\mathrm{i}$. Then, the RCS energy constraint is denoted as:

\begin{equation}
\label{PSOC2}
\sum_{k \in \Omega_\mathrm{i}} e_\mathrm{k} \leq E_\mathrm{i}, \forall R_\mathrm{i} \in \mathcal{R}.
\end{equation}

\noindent Therefore, the price decision process is formulated as an optimization problem as:

 \begin{align}
 \label{Op2}
\max_{p} & \quad U_\mathrm{P}(e_\mathrm{k},p) \\
 \textrm{s.t.}  & \quad (\ref{PSOC1}), (\ref{PSOC2}) \notag.
\end{align}

\section{Game Analysis}

\subsection{Existence and Uniqueness of Stackelberg Game}

By solving problem (\ref{Op1}), we can obtain the best-response strategy of MES $k$, denoted as $e^{\star}_\mathrm{k}(p)$. When followers are at Nash equilibrium, all followers choose their best-response strategies and the strategy set is denoted as \begin{math} \{e^{\star}_\mathrm{k}(p)\}_\mathrm{k\in \mathcal{K}}=\{e^\star_\mathrm{1}(p), \dots, e^\star_\mathrm{K}(p)\} \end{math}. Given MES best-response strategy profile, the optimal price $p^\star$ can be obtained by solving problem (\ref{Op2}). Therefore, the profile of $(p^\star, \{e^{\star}_\mathrm{k}(p)\}_\mathrm{k\in \mathcal{K}})$ is the Stackelberg equilibrium for the proposed game, which is calculated as:

 \begin{align}
 \label{Op3}
(p^\star, \{e^{\star}_\mathrm{k}(p)\}_\mathrm{k\in \mathcal{K}})= & \argmax_{p} U_\mathrm{P}(p, \{e^\star_\mathrm{k}(p)\}_\mathrm{k\in \mathcal{K}}) \\
 \textrm{s.t.}  & \quad e^\star_\mathrm{k}(p)=\argmax U_\mathrm{k}(e_\mathrm{k},p), \quad k \in \mathcal{K}\notag.
\end{align}

We first analyze the follower-level game by computing MES best-response strategy in the following lemma.

\begin{lemma}
MES $k$ has a unique best-response strategy $e^\star_\mathrm{k}(p)$ for a given service price $p$, denoted as:

\begin{equation}
\label{BRF}
e^\star_\mathrm{k}(p)=
  \begin{cases}
    0,     & \quad p \leq p^L_\mathrm{k} \\
    \frac{2p-\alpha^T_\mathrm{k}\frac{P_\mathrm{i}+P_\mathrm{j}}{P_\mathrm{i}P_\mathrm{j}}-\alpha^D_\mathrm{k}\alpha_\mathrm{2}D_\mathrm{k}/B_\mathrm{k}}{2\alpha_\mathrm{1}\alpha^D_\mathrm{k}D_\mathrm{k}/B^2_\mathrm{k}},  & \quad p^L_\mathrm{k} < p < p^U_\mathrm{k} \\
    B_\mathrm{k}-e^I_\mathrm{k}, & \quad p \geq p^U_\mathrm{k}
  \end{cases}
\end{equation}

\noindent where $p^L_\mathrm{k}$ is the rejection price, below which MES $k$ will not provide service. $p^U_\mathrm{k}$ is the saturated price at which MES $k$ provides maximal service capacity. $p^L_\mathrm{k}$ and $p^U_\mathrm{k}$ are denoted as:

\begin{equation}
\label{MESprice}
\scriptstyle
\left\{
             \begin{array}{lr}
             p^L_\mathrm{k}=0.5(\alpha^T_\mathrm{k}\frac{P_\mathrm{i}+P_\mathrm{j}}{P_\mathrm{i}P_\mathrm{j}}+\alpha^D_\mathrm{k}\alpha_\mathrm{2}D_\mathrm{k}/B_\mathrm{k}) \\
			p^U_\mathrm{k}= p^L_\mathrm{k}+\alpha_\mathrm{1}\alpha^D_\mathrm{k}D_\mathrm{k}(B_\mathrm{k}-e^I_\mathrm{k})/B^2_\mathrm{k}
             \end{array}
\right.
\end{equation}

\end{lemma} 

\begin{proof}
For MES $k$, the strategy set is denoted as $\{e_\mathrm{k}|e_\mathrm{k}\in R, 0\leq e_\mathrm{k}\leq B_\mathrm{k}-e^I_\mathrm{k} \}$, which is the intersection of two half-spaces. Thus, the MES strategy set is non-empty and convex. To find the best-response strategy of $k$-th MES, we solve the optimization problem (\ref{Op1}). First, we analyze the property of the objective function $U_\mathrm{k}(e_\mathrm{k},p)$ by calculating the second-order derivative of the function:

\begin{equation}
\label{deri}
\frac{\partial^2 U_\mathrm{k}(e_\mathrm{k},p)}{\partial e_\mathrm{k}^2}=-2\alpha^D_\mathrm{k}\alpha_\mathrm{1}D_\mathrm{k}.
\end{equation}

\noindent As \begin{math}\alpha^D_\mathrm{k}, \alpha_\mathrm{1},  D_\mathrm{k}>0\end{math}, the value of equation (\ref{deri}) is negative. Thus, problem (\ref{Op1}) is proven to be a convex optimization problem, and the best response strategy for MES $k$ is the global optimum. By applying Lagrangian function and Karush-Kuhn-Tucker (KKT) conditions to problem (\ref{Op1}), we can obtain the best response strategy of MES $k$. The detailed calculation is omitted due to space limitations.

\end{proof}


Based on the MES best-response strategy,  we define the feasible range of service price $p$. The PSO adjusts its price within the range between the minimal and maximal value of \begin{math} p_\mathrm{range} \triangleq [p^L_\mathrm{1} \dots p^L_\mathrm{K}, p^U_\mathrm{1} \dots p^U_\mathrm{K}]\end{math}. When the price is below \begin{math} \min\{p_\mathrm{range}\}\end{math}, all MESs will not participate in the service. When the price reaches \begin{math} \max\{p_\mathrm{range}\}\end{math}, all MESs will use up their battery space for the service and no higher price is needed. By calculating $p^L_\mathrm{k}$ and $p^U_\mathrm{k}$ for MES $k$ using equation (\ref{MESprice}), and sorting all $p^L_\mathrm{k}$ and $p^U_\mathrm{k}$ in an ascending order, we have the feasible set of the price. The price set is an $M$-element vector $\gamma$, where \begin{math} \gamma_\mathrm{1} \leq \gamma_\mathrm{2} \leq \dots \leq \gamma_\mathrm{m} \leq \gamma_\mathrm{m+1} \dots \gamma_\mathrm{M} \end{math}. Further, define \begin{math} \Gamma_\mathrm{m} \triangleq [\gamma_\mathrm{m}, \gamma_\mathrm{m+1}] \end{math} for \begin{math} m=1, 2, \dots M-1 \end{math}, we can divide the price $p$ range into $M-1$ intervals.

To find the optimal price for PSO, we decompose problem (\ref{Op2}) into $M-1$ sub-problems  where the $m$-th sub-problem aims to find the optimal price within the range $\Gamma_\mathrm{m}$, similar to the work in \cite{mushubargain}.

\begin{lemma}
In the sub-domain of \begin{math} \Gamma_\mathrm{m}, \forall m \end{math}, problem (\ref{Op2}) is a convex optimization problem.
\end{lemma}

\begin{proof}
As the price is continuous within $\Gamma_\mathrm{m}$, the price set is convex. By substituting the MES best-response strategy into problem (\ref{Op2}), the objective function is calculated as:

\begin{multline}
\label{PSOutility2}
U_\mathrm{P}( p, \Gamma_\mathrm{m})=\\
\alpha_\mathrm{L}\sum_{L_\mathrm{j}\in \mathcal{L}}(-(a_\mathrm{j}(\sum_{k \in \Lambda_\mathrm{j}\cap \phi_\mathrm{1}}(y_\mathrm{k}p-z_\mathrm{k})+\sum_{k \in \Lambda_\mathrm{j}\cap \phi_\mathrm{2}}(B_\mathrm{k}-e^I_\mathrm{k}))-b_\mathrm{j})^2+c_\mathrm{j}) \\
-2p(\sum_{k \in \phi_1}(y_\mathrm{k}p-z_\mathrm{k})+\sum_{k \in \phi_2} (B_\mathrm{k}-e^I_\mathrm{k}))+\sum_{k\in \mathcal{K}}p\bar{e},
\end{multline}

\noindent where \begin{math} y_\mathrm{k}p-z_\mathrm{k} \end{math} is the simplied function of $e_\mathrm{k}^\star$ for \begin{math}p^L_\mathrm{k} < p < p^U_\mathrm{k} \end{math}. For MESs with non-zero best-response value within $\Gamma_\mathrm{m}$, they are divided into two sets: $\phi_\mathrm{1}$ and $\phi_\mathrm{2}$, where \begin{math} \phi_\mathrm{1}=\{k|y_\mathrm{k}p-z_\mathrm{k}<B_\mathrm{k}-e^I_\mathrm{k}\} \end{math} and \begin{math} \phi_\mathrm{2}=\{k|y_\mathrm{k}p-z_\mathrm{k}\geq B_\mathrm{k}-e^I_\mathrm{k}\} \end{math}. As $p^L_\mathrm{k}$ and $p^U_\mathrm{k}$ are deterministic and irrelevant to $p$, $\phi_\mathrm{1}$ and $\phi_\mathrm{2}$ are deterministic and fixed.  The second derivative of the utility function is calculated as:

\begin{equation}
\frac{\partial^2 U_\mathrm{P}(p, \Gamma_\mathrm{m})}{\partial p^2}=-2\alpha_\mathrm{L} \sum_{L_\mathrm{j}\in \mathcal{L}}a_\mathrm{j}^2 \sum_{k \in \Lambda_\mathrm{i}\cap \phi_\mathrm{1}}y_\mathrm{k}^2-4 \sum_{k \in \phi_1}y_\mathrm{k},
\end{equation}

\noindent where \begin{math} \alpha_\mathrm{L}, a_\mathrm{j}, y_\mathrm{k}>0 \end{math}. Thus, \begin{math} \partial^2 U_\mathrm{P}(p, \Gamma_\mathrm{m})/\partial p^2<0 \end{math}, making the utility function concave and differential. Moreover, by substituting equation (\ref{MESprice}) to constraints (\ref{PSOC1}), (\ref{PSOC2}), both constraints are convex (half-space). Therefore, problem (\ref{Op2}) within the sub-domain $\Gamma_\mathrm{m}$ is a convex optimization problem.

\end{proof}

\begin{lemma}
The PSO has a globally optimal price, given the best-response strategies of MESs.
\end{lemma}

\begin{proof}
By decomposing problem (\ref{Op2}) into $M-1$ sub-problems as defined in Lemma 2, the original problem can be rewritten as:

 \begin{align}
 \label{Op22}
\max_{m} \max_{p \in \Gamma_\mathrm{m}} & \quad U_\mathrm{P}(p, \Gamma_\mathrm{m}) \\
 \textrm{s.t.}  & \quad (\ref{PSOC1}), (\ref{PSOC2}) \notag.
\end{align}

\noindent By obtaining the optimal result $p_\mathrm{m}^\star$ of the convex sub-problem following Lemma 2, we can find the globally optimal price $p^\star$ by searching the maximum utility value from $[p_\mathrm{1}^\star, \dots, p_\mathrm{m}^\star, \dots, p_\mathrm{M-1}^\star]$:

\begin{align}
\label{Op23}
p^\star=\argmax_{p\in [p_\mathrm{1}^\star, \dots, p_\mathrm{m}^\star, \dots, p_\mathrm{M-1}^\star]} & \quad U_\mathrm{P}(p, \Gamma_\mathrm{m}) \\
 \textrm{s.t.}  & \quad (\ref{PSOC1}), (\ref{PSOC2}) \notag.
\end{align}

\end{proof}

Thus, the existences and uniqueness of Stackelberg equilibrium can be proved in the following proposition.

\begin{prop}
For the formulated game, a unique Stackelberg equilibrium exists.
\end{prop}

\begin{proof}
As shown in Lemma 1, each MES has a unique best-response strategy $e^\star_\mathrm{k}(p)$ given a posted price $p$ . Then, by substituting $e^\star_\mathrm{k}(p)$ to problem (\ref{Op2}), we prove the global optimum of PSO strategy as in Lemma 2 and Lemma 3. As the PSO achieves its global optimum and each MES has a unique best-response strategy, the unique Stackelberg equilibrium is obtained.
\end{proof}

\subsection{Stackelberg Game Algorithm}

During each scheduling time slot $h$, the PSO scheduling price can be obtained as in Algorithm \ref{algo1}.

\begin{algorithm}
\footnotesize
\label{algo1}

\For{k=1 to K}{
Calculate $p^L_\mathrm{k}$ and $p^L_\mathrm{k}$ according to equation (\ref{MESprice}) \;
}
Sorting $p^L_\mathrm{k}$ and $p^L_\mathrm{k}$ in an ascending order to form vectors $\gamma$ and $\Gamma$ \linebreak
\For{m=1 to M-1}{
Find $p_\mathrm{m}^\star$ by solving problem (\ref{Op2}) within $[\Gamma_\mathrm{m},\Gamma_\mathrm{m+1}]$\;
}
Find the optimal $p^\star$ with the maximal utility value according to equation (\ref{Op23}).

 \caption{Stackelberg game solution.}
\end{algorithm}

The proposed algorithm does not require iterations to analyze the Stackelberg game. As the number of MESs increases, the algorithm complexity increases accordingly. Thus, the proposed algorithm can be applied to schedule a large number of MESs.

\section{Simulation Results}

To validate the effectiveness of the proposed scheme, we present simulation results in this section. The parameter setting is shown in Table \ref{table1}. RCSs $R_\mathrm{1}$ has 1.6MWh and $R_\mathrm{2}$ has 900kWh surplus energy for MESs to deliver. LCS $L_\mathrm{1}$ demands 100-200kWh energy to be delivered while $L_\mathrm{2}$ demands energy between 150-300kWh. RCSs adopt society of automotive engineers (SAE) combined charging system (CCS) level 2 charging standard at 90kW and LCSs adopt SAE CCS charging standard at 60kW \cite{SAECCS}. The number of MESs $N_\mathrm{ij}$ along the $R_\mathrm{i}$-$L_\mathrm{j}$ pair is also included in Table \ref{table1}. The MES service capacities are considered as random variables that follow a normal distribution with a mean value of 14 and a standard deviation of 5 (kWh). Similarly, the MES battery capacities also follow a normal distribution with an 80kWh mean and a standard deviation of 10kWh. The battery degradation parameters $D_\mathrm{k}$, $\alpha_\mathrm{1}$, $\alpha_\mathrm{2}$, $\alpha_\mathrm{3}$ are calculated according to data in the work \cite{batterydegradation26}. While the MES battery degradation cost is relatively low (\emph{e.g.}, 4$\times 10^{-4}$) per cycle, it is still a great concern for MES drivers. Thus, $\alpha^D_\mathrm{k}$ is set to 10$^5$ and $\alpha^T_\mathrm{k}$ is set to 30 to make them comparable with service reward and motivation reward. The loading cost weight $\alpha_\mathrm{L}$ is set to 0.5.

\begin{table}

\setlength{\abovecaptionskip}{-5pt}
\scriptsize

\renewcommand\arraystretch{1.5}

\caption{Simulation Parameters}
\label{table1}

\begin{center}
\begin{tabular}{p{40pt} p{50pt} | p{40pt} p{50pt}  }
	\hline
	Para. & Value & Para. & Value \\
	\hline
	$E_\mathrm{1}$ & 1.6MWh & $E_\mathrm{2}$ & 900kWh \\
	\hline
	$D^L_\mathrm{1}$& 100kWh & $D^U_\mathrm{1}$ & 200kWh \\
	\hline
	$D^L_\mathrm{2}$ & 150kWh & $D^U_\mathrm{2}$& 300kWh \\
	\hline
	$P_\mathrm{R_\mathrm{1}/R_\mathrm{2}}$ & 90kW & $P_\mathrm{L_\mathrm{1}/L_\mathrm{2}}$ & 60kW \\
	\hline
	$N_\mathrm{11}$ & 6 & $N_\mathrm{12}$ & 8 \\
	\hline
	$N_\mathrm{21}$ & 7 & $N_\mathrm{22}$ & 4 \\
	\hline
	E(${e_\mathrm{k}}$) & 14kWh &$\sigma({e_\mathrm{k}})$ & 5kWh \\
	\hline
	E(${B_\mathrm{k}}$) & 80kWh & $\sigma({B_\mathrm{k}})$ & 10kWh \\
	\hline
	$\alpha^T_\mathrm{k}$ & 30 & $\alpha^D_\mathrm{k}$ & 10$^5$ \\
	\hline
	$D_\mathrm{k}$ & 5.08$\times$10$^{-4}$ & $\alpha_\mathrm{1}$ & 1 \\
	\hline
	$\alpha_\mathrm{2}$ & -0.222 & $\alpha_\mathrm{L}$ & 0.5 \\ 
	\hline

\end{tabular}
\end{center}
\end{table}

We compare the proposed scheme with the price-minimized scheme and random scheme in terms of PSO utility revenue, as shown in Fig. \ref{simulation1}. In the price-minimized scheme, the PSO aims to minimize its service price. In the random scheme, the PSO randomly adjusts the service price to meet LCS energy demands. It can be seen that the proposed scheme has the highest utility revenue compared with both price-minimized and random schemes as the price-minimized scheme only tries to minimize the price, but ignores the loading revenue impact on the utility function. Compared to the price-minimized scheme, the random scheme schedules more MESs and can achieve higher revenue as it does not put strict constraints to achieve minimal loading demands. Moreover, as MES service capacities increase, the PSO utility increases. Since there are more on-road service capacities, a lower service price can be posted and the PSO can have a high loading revenue. For the price-minimized scheme, the revenue increment is smaller as it provides minimal LCS loading demands and the loading revenue stays almost the same.

\begin{figure}[!t]
\centering
\setlength{\abovecaptionskip}{-5pt}
\includegraphics[width=0.5\textwidth]{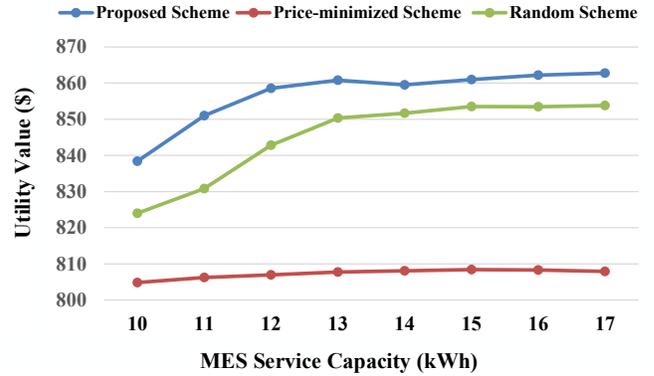}
\vspace{-0.3cm}
\caption{PSO utility revenue with different MES service capacities.}
\vspace{-0.3cm}
\label{simulation1}
\end{figure}

 \begin{figure}[!t]
\centering
\setlength{\abovecaptionskip}{-5pt}
\includegraphics[width=0.5\textwidth]{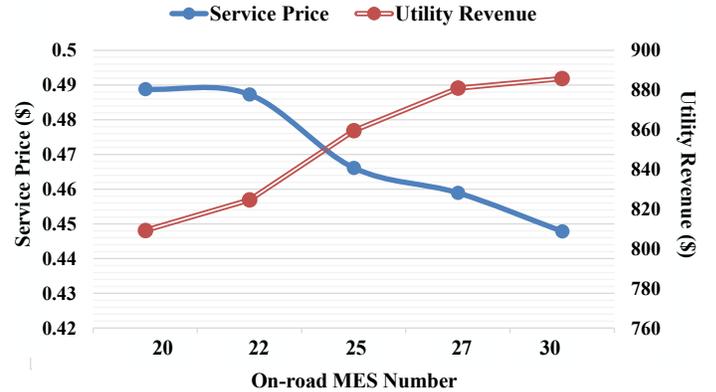}
\vspace{-0.3cm}
\caption{Impact of on-road MES number on energy scheduling.}
\vspace{-0.3cm}
\label{simulation2}
\end{figure}

We also discuss the energy scheduling scheme under different operation scenarios. The impact of on-road MES number on the scheduling result is shown in Fig. \ref{simulation2}. It can be seen that as the on-road MES number increases, the service price decreases since the PSO has more potential energy servers, and less motivation is required. Correspondingly, the PSO utility revenue increases. Moreover, with more MESs participating in the service, more energy can be delivered to LCSs, which increases the loading revenue part of the utility.

 \begin{figure}[!t]
\centering
\setlength{\abovecaptionskip}{-5pt}
\includegraphics[width=0.5\textwidth]{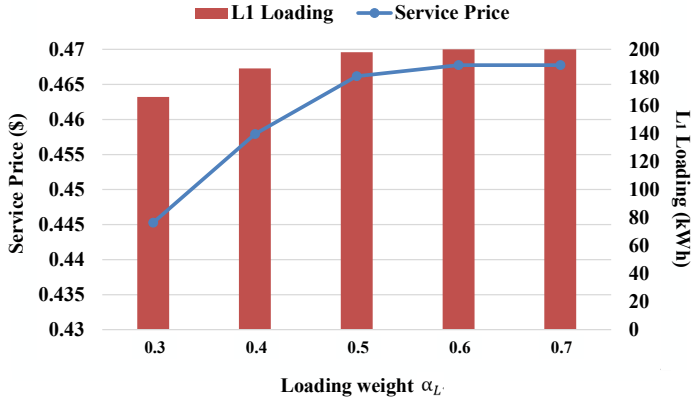}
\vspace{-0.3cm}
\caption{Impact of the loading weight on energy scheduling.}
\vspace{-0.3cm}
\label{simulation3}
\end{figure}

Depending on the PSO operation goal, the MES scheduling could lean towards operation cost minimization or loading revenue maximization. By adjusting the loading weight $\alpha_\mathrm{L}$, the operation objectives vary, and the scheduling result also changes, as shown in Fig. \ref{simulation3}. It can be seen that as the loading weight increases, the MES scheduling mainly focuses on loading revenue maximization. To encourage MESs delivering more energy, the service price increases until $\alpha_\mathrm{L}$ reaches 0.6. We can observe from the figure that when \begin{math}  \alpha_\mathrm{L}=0.6\end{math}, the loading at $L_\mathrm{1}$ reaches its maximal demanding load $D^U_\mathrm{1}$. Therefore, a higher loading weight will result in the same loading results as limited by the loading constraints, and the service price will remain the same.

 \begin{figure}[!t]
\centering
\setlength{\abovecaptionskip}{-5pt}
\includegraphics[width=0.5\textwidth]{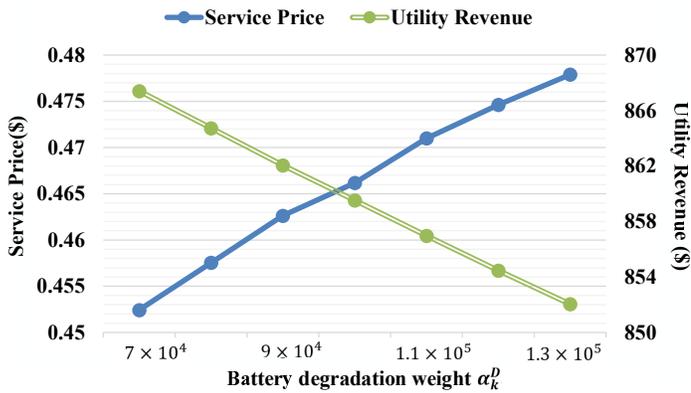}
\vspace{-0.3cm}
\caption{Impact of the battery degradation weight on energy scheduling.}
\vspace{-0.3cm}
\label{simulation4}
\end{figure}

As the battery technology advances, the MES driver's preference towards energy discharging also changes. Therefore, we discuss the battery degradation weight impact on scheduling results, as shown in Fig. \ref{simulation4}. As the weight $\alpha^D_\mathrm{k}$ increases, MES drivers are more reluctant to discharge energy, and the price range $p_\mathrm{range}$ becomes wider. Therefore, to stimulate MESs actively fulfilling the tasks, PSO needs to post a higher service price. As a result, the operation cost increases, and the PSO utility revenue decreases accordingly.

\section{Conclusion}

In this paper, a price-incentive scheme that stimulates MESs fulfilling energy compensation tasks has been proposed to mitigate overload issues at LCSs. The interaction between the PSO and on-road MESs has been formulated as a Stackelberg game. The existence and uniqueness of Stackelberg equilibrium have been proven, and an algorithm has been designed to find the equilibrium. Simulation results have validated the effectiveness of the proposed scheme under different operation scenarios. The proposed scheme can be applied by the local PSO to balance the system energy without excessive power infrastructure upgrade while MESs are stimulated to fulfill tasks in a cost-efficient way.

For our future works, we will consider unexpected errors during MES service (\emph{e.g.}, human behaviour, loading change, etc.) to enable robust energy scheduling.

\vfill\eject


\begin{thebibliography}{99}
\bibliographystyle{plain} 

\bibitem{expectation}
K. Lindquist and M. Wendt, ``Electric vehicle policies, fleet, and infrastructure: Synthesis", Washington State Department of Transportation, Tacoma, WA, Tech. Rep., Nov. 2011.

\bibitem{impact2}
R. Leou, C. Su and C. Lu, ``Stochastic analyses of electric vehicle charging impacts on distribution network,'' \emph{IEEE Trans. Power Syst.}, vol. 29, no. 3, pp. 1055-1063, May 2014.

\bibitem{bidirectional}
V. Monteiro, J. C. Ferreira, A. A. N. Mel\'{e}ndez, C. Couto and J. L. Afonso, ``Experimental validation of a novel architecture based on a dual-stage converter for off-board fast battery chargers of electric vehicles'', \emph{IEEE Trans. Veh. Technol.}, vol. 67, no. 2, pp. 1000-1011, Feb. 2018.

\bibitem{nanTVT}
N. Chen, J. Ma, M. Wang and X. Shen, ``Two-tier energy compensation framework based on mobile vehicular electric storage," \emph{IEEE Trans. Veh. Technol.}, vol. 67, no. 12, pp. 11719-11732, Dec. 2018.

\bibitem{pyi}
P. Yi, Y. Tang, Y. Hong, Y. Shen, T. Zhu, Q. Zhang and M. M. Begovic, ``Renewable energy transmission through multiple routes in a mobile electrical grid," in \emph{2014 IEEE PES Innovative Smart Grid Technologies (ISGT)}, Washington, DC, 2014, pp. 1-5.

\bibitem{DSMrobustness}
W. Zhong, R. Yu, S. Xie, Y. Zhang and D. K. Y. Yau, ``On stability and robustness of demand response in V2G mobile energy networks," \emph{IEEE Trans. Smart Grid}, vol. 9, no. 4, pp. 3203-3212, July 2018.


\bibitem{approach}
H. Yang, X. Xie and A. V. Vasilakos, ``Noncooperative and cooperative optimization of electric vehicle charging under demand uncertainty: A robust Stackelberg game," \emph{IEEE Trans. Veh. Technol.}, vol. 65, no. 3, pp. 1043-1058, Mar. 2016.

\bibitem{2dmarkov}
I. S. Bayram, G. Michailidis and M. Devetsikiotis, ``Unsplittable load balancing in a network of charging stations under QoS guarantees," \emph{IEEE Trans. Smart Grid}, vol. 6, no. 3, pp. 1292-1302, May 2015.

\bibitem{multicharging}
W. Yuan, J. Huang and Y. J. A. Zhang, ``Competitive charging station pricing for plug-in electric vehicles," \emph{IEEE Trans. Smart Grid}, vol. 8, no. 2, pp. 627-639, Mar. 2017.

\bibitem{v2gglobecom}
X. Chen and K. Leung, ``A game theoretic approach to vehicle-to-grid scheduling," in \emph{IEEE Global Communication Conference (GLOBECOM)}, Abu Dhabi, United Arab Emirates, 2018, pp. 1-6.


\bibitem{batterydegradation}
A. Ahmadian, M. Sedghi, B. Mohammadi-ivatloo, A. Elkamel, M. Aliakbar Golkar and M. Fowler, ``Cost-benefit analysis of V2G implementation in distribution networks considering PEVs battery degradation," \emph{IEEE Trans. Sustainable Energy}, vol. 9, no. 2, pp. 961-970, Apr. 2018.


\bibitem{mushubargain}
M. Li, J. Gao, L. Zhao and X. Shen, ``Task time allocation and reward scheme for PEV charging station advertising," in \emph{2019 IEEE International Conference on Communications (ICC)}, Shanghai, China, May 20-24, 2019.

\bibitem{SAECCS}
\emph{IEEE Standard Technical Specifications of a DC Quick Charger for Use with Electric Vehicles}, IEEE Std 2030.1.1-2015 Std., 2016.

\bibitem{batterydegradation26}
C. Guenther, B. Schott, W. Hennings, P. Waldowski and M. A. Danzer, ``Model-based investigation of electric vehicle battery aging by means of vehicle-to-grid scenario simulations," \emph{J. Power Sources}, vol. 239, no. 1, pp. 604-610, Oct. 2013.

\end{thebibliography}
\end{document}